\documentclass[11pt]{llncs}
\usepackage[a4paper]{geometry}
\usepackage{amssymb}
\usepackage{amsfonts}
\usepackage{amsmath}
\usepackage{mathrsfs}
\usepackage{latexsym}
\usepackage{color}
\usepackage[utf8]{inputenc}
\usepackage{verbatim}
\newtheorem{observation}{Observation}

\newcommand{\OPT}{\mathrm{OPT}}
\newcommand{\UG}{\mathcal{U}^\Gamma}

\newcommand{\op}{\overline{p}}
\newcommand{\hp}{\hat{p}}
\renewcommand{\P}{{\cal P}}
\newcommand{\NP}{{\cal NP}}

\newcommand{\Q}{\mathbb{Q}}

\newcommand{\barp}{\overline p}
\newcommand{\hatp}{\hat p}
\newcommand{\T}{\mathcal{P}}
\newcommand{\rT}{\overline{\mathcal{P}}}
\newcommand{\m}{\overline{m}}
\renewcommand{\o}{o}
\newcommand{\ro}{\overline{o}}

\title{Approximation results for makespan minimization with budgeted uncertainty \thanks{This work was partially supported by DFG Project, “Robuste Online-Algorithmen für Scheduling-und Packungsprobleme”, JA 612/19-1, and ANR project ROBUST (ANR-16-CE40-0018).}}
\author{Marin Bougeret$^2$, Klaus Jansen$^1$, Michael Poss$^2$, and Lars Rohwedder$^1$}
 \date{
     \today
 }

\institute{%
     $^1$Department of Computer Science, Kiel University, 24098 Kiel, Germany\\%
  $^2$LIRMM, University of Montpellier, CNRS, France\\[2ex]
  \{marin.bougeret,michael.poss\}@lirmm.fr \\
  \{kj,lro\}@informatik.uni-kiel.de}

\begin{document}

 \maketitle

  \begin{abstract}
   We study approximation algorithms for the problem of minimizing the makespan on a set of machines with uncertainty on the processing times of jobs.
   In the model we consider, which goes back to~\cite{BertsimasS03}, once the schedule is defined an adversary can pick a scenario where deviation
   is added to some of the jobs' processing times.
   Given only the maximal cardinality of these jobs, and the magnitude of potential deviation for each job, the goal is to optimize the worst-case scenario.
   We consider both the cases of identical and unrelated machines. Our main result is an EPTAS for the case of identical machines.
   We also provide a $3$-approximation algorithm and an inapproximability ratio of $2-\epsilon$ for the case of unrelated machines.
   
\end{abstract}

   \section{Introduction}

   Classical optimization models suppose perfect information over all parameters. This can lead to optimal solutions having poor performance when the actual parameters deviate, even by a small amount, from the predictions used in the optimization model. Different frameworks have been proposed to overcome this issue, among which Robust Optimization which tackles the uncertainty by providing a set of possible values for these parameters, and considering the worst outcome over that set. In this paper, we consider the problem of scheduling a set of jobs $J$ on the set of machines $M$, so as to minimize the makespan, and considering that the processing times are uncertain. What is more, we consider the budgeted uncertainty model introduced by~\cite{BertsimasS03} where each processing time varies between its nominal value and the latter plus some deviation. Further, in any scenario, at most $\Gamma$ of the uncertain parameters take the higher values, the other being at their nominal values.

Let us now formally define the Robust Scheduling on Unrelated Machines ($R|\UG|C_{\max}$) problem.
For any job $j \in J$ and machine $i \in M$, we denote by $\overline p_{ij} \ge 0$ the nominal processing time of $j$ on $i$, and by $\hat p_{ij} \ge 0$ the (potential) deviation of $j$ on $i$. 
A schedule $\sigma$ is a function from $J\rightarrow M$. We write $\sigma_i$ for the subset of jobs scheduled on machine $i$.
Let $\UG=\{\xi \in \{0,1\}^{|J|}:\|\xi\|_1\leq \Gamma\}$ be the set of all possible scenarios where at most $\Gamma$ jobs deviate.
For any $\xi \in \UG$, we set $p^{\xi}_{ij}=\overline p_{ij} + \xi_{j} \hat p_{ij}$ to be the actual processing time of $j$ on $i$ in scenario $\xi$.

Let us now formalize some common terms, but with dependence on scenario $\xi$. The load of machine $i$ in scenario $\xi$ is calculated as $\sum_{j\in\sigma_i} p^{\xi}_{ij}$.
The makespan in scenario $\xi$ is the maximum load in scenario $\xi$, i.e., $C^{\xi}_{\max}(\sigma) = \max_{i\in M} \sum_{j\in\sigma_i} p^{\xi}_{ij}$.
Finally, $C_{\max}^{\Gamma}(\sigma) = \max_{\xi\in \UG} C_{\max}^{\xi}(\sigma)$ denotes the objective function we consider in Robust Scheduling, where the adversary takes the worst scenario among $\UG$.


Next, we will state important observations about the objective function.
We first need to introduce the following notations. Given a set of jobs $X_i$ scheduled on machine $i$, we define $\barp(X_i)= \sum_{j\in X_i}\barp_{ij}$,
$\hatp(X_i)= \sum_{j\in X_i}\hatp_{ij}$, $\Gamma(X_i)$ as the set of the $\Gamma$ jobs of $X_i$ with the largest $\hatp_{ij}$ values (or $\Gamma(X_i)=\sigma_i$ when $|X_i| < \Gamma$) with ties broken arbitrarily.
Finally, set $\hatp_{\Gamma}(X_i) = \hatp(\Gamma(X_i))$.

By definition we have
$C_{\Gamma}(\sigma)=   \max_{\xi\in \UG} \max_{i\in M} \sum_{j\in\sigma_i} C_{\xi}(\sigma)$, and thus we can rewrite 
$C_{\Gamma}(\sigma)=   \max_{i\in M} \max_{\xi\in \UG} \sum_{j\in\sigma_i} C_{\xi}(\sigma) =  \max_{i\in M} C_{\Gamma}(\sigma_i)$, where
$C_{\Gamma}(\sigma_i) = \max_{\xi\in \UG} \sum_{j\in\sigma_i} C_{\xi}(\sigma)$ is the worst-case makespan on machine $i$.
The benefit of rewriting $C_{\Gamma}(\sigma)$ in this form is that it is now clear that $C_{\Gamma}(\sigma_i) = \barp(\sigma_i)+\hatp_{\Gamma}(\sigma_i)$
as the worst scenario $\xi$ for a fixed $\sigma_i$ is obtained by picking the $\Gamma$ jobs with highest $\hat p_{ij}$ and make them deviate.
Thus, $R|\UG|C_{\max}$ can also be thought as a ``classical'' scheduling problem (without adversary) where the makespan on a machine  $C_{\Gamma}(\sigma_i)$
is simply the sum of all the nominal processing time of jobs of $\sigma_i$, plus only the $\Gamma$ largest deviating values of jobs of $\sigma_i$.
We are now ready to define $R|\UG|C_{\max}$.

\begin{problem}{\sc Robust Scheduling on Unrelated Machines ($R|\UG|C_{\max}$)}
\label{pb:rugcmax}
  \begin{itemize}
  \item Input: $(J, M, \op\in\Q^{|M|\times |J|}_+, \hp\in\Q^{|M|\times |J|}_+)$ where $J$ is the set of jobs, $M$ the set of machines,
    $\op$ are the vectors of nominal processing times, and $\hp$ the vectors of deviation
  \item Output: find a schedule $\sigma:J\rightarrow M$
  \item Objective function: min $C_{\Gamma}(\sigma)=   \max_{\xi\in \UG} \max_{i\in M} \sum_{j\in\sigma_i}[\overline p_{ij} + \xi_{j} \hat p_{ij}] =  \max_{i\in M} C_{\Gamma}(\sigma_i)$,
    where $C_{\Gamma}(\sigma_i) = \barp(\sigma_i)+\hatp_{\Gamma}(\sigma_i)$.
  \end{itemize}
\end{problem}


Following the classical three field notation, we denote by $R|\UG|C_{\max}$ the previous problem.
Notice that when all $\hp_{ij}=0$ the problem corresponds to the classical $R||C_{\max}$, for which we denote by $C(\sigma_i)=\sum_{j \in \sigma_i}\barp_{ij}$ the makespan on machine $i$. 
We are also interested in simplifications of the above problem. The first one is {\sc Robust Scheduling on Uniform Machines ($Q|\UG|C_{\max}$)} where
each machine has a speed $s_i$, each job has two processing times ($\barp_j$ and $\hp_j$), and we have $\barp_{ij}=\frac{\barp_j}{s_i}$
and $\hatp_{ij}=\frac{\hatp_j}{s_i}$.
The other simplification is {\sc Robust Scheduling on Identical Machines ($P|\UG|C_{\max}$)} where $s_i = 1$ for any machine $i$.

Robust scheduling has been considered in the past, mostly for finite uncertainty sets without particular structure, see for instance~\cite{Aloulou2008338,DanielsKouvelis1995,KasperskiKZ12,KasperskiKZ12b,Mastrolili08}. More recently, \cite{BougeretPP18,TadayonS15,binpacking19} considered robust packing and scheduling with the budgeted uncertainty model $\UG$ from~\cite{BertsimasS04}. Specifically,~\cite{BougeretPP18} (with authors in common) provided a 3-approximation algorithm and a $(1+\epsilon)$-approximation (PTAS) for $P|\UG|C_{\max}$ but only for a constant $\Gamma$, as well as a randomized approximation algorithm for $R|\UG|C_{\max}$ having an average ratio of $O(\log(m))$. They also considered problem $1|\UG|\sum_jw_jC_{\max}$, proving that the problem is $\NP$-hard in the strong sense, and providing a polynomial-time algorithm when $w_j=1$ for $j\in J$. Authors of~\cite{TadayonS15} considered the robust one-machine problem for four commonly-used objective criteria: (weighted) total completion time, maximum lateness/tardiness, and number of late jobs. They showed that some of these problems are polynomially solvable and provide mixed-integer programming formulations for others. Their results considered $\UG$ as well as two closely related uncertainty sets. Paper~\cite{binpacking19} (with also authors in common) considers robust bin-packing problem for $\UG$ and one of the uncertainty sets considered by~\cite{TadayonS15}, and provided constant-factor approximations algorithms for the two problems.

In this paper we improve the results of \cite{BougeretPP18} for $P|\UG|C_{\max}$ and $R|\UG|C_{\max}$.
In Section~\ref{sec:3approx} we show that any $c$-approximation for the classical
$R||C_{\max}$ problem leads to a $(c + 1)$-approximation for $R|\UG|C_{\max}$, hence obtaining a $3$-approximation algorithm for the latter problem, and
a  $(2 + \epsilon)$-approximation for $Q|\UG|C_{\max}$. We point out that this result improves the \textit{ad-hoc} $3$-approximation of~\cite{BougeretPP18} for $P|\UG|C_{\max}$,
while having a simpler proof.
In Section~\ref{sec:inapprox}, we show through a reduction from the {\sc Restricted Assignment Problem} that there exists no $(2-\epsilon)$-approximation algorithm for $R|\UG|C_{\max}$ unless $\P=\NP$.
In Section~\ref{sec:PTAS}  we consider the $P|\UG|C_{\max}$ problem and present the first step our main result, namely a PTAS which is valid even when $\Gamma$ is part of the input, i.e., not constant.
Having $\Gamma$ in the input (and not constant) requires a totally different technique from the one used in \cite{BougeretPP18}.
The algorithm is turned into an EPTAS in Section~\ref{sec:EPTAS}, i.e., a PTAS where the dependency
of $\epsilon$ is not in the exponent of the encoding length.






  \section{A 3-approximation for unrelated machines}
  \label{sec:3approx}
  
    \begin{theorem}
      \label{th-main}
      Any polynomial $c$-approximation for $R||C_{\max}$ implies a polynomial $(c + 1)$-approximation for $R|\UG|C_{\max}$.
    \end{theorem}
    \begin{proof}[Proof of Theorem~\ref{th-main}]
      We design a dual approximation, i.e., given an instance $I$ of $R|\UG|C_{\max}$ and an threshold $T$,
      we either give a schedule $\sigma$ of $I$ with $C_{\Gamma}(\sigma) \le (c + 1) T$, or prove that $T < \OPT(I)$.
Using a binary search on $T$ this will imply a $(c+1)$-approximation algorithm.
      
      For that, given an instance $I=(J, M, \op, \hp)$ of $R|\UG|C_{\max}$, and $T$ the current threshold, 
      our objective is to define an instance $I' =(J, M, p)$ of the classical  $R||C_{\max}$ problem.
      The transformation of a solution for $I'$ to a solution for $I$ will be straightforward since the jobs and machines will be the same.

      Given a machine $i$, let $B_i = \{j | \hat p_{ij} > \frac T \Gamma \}$ and $S_i = J \setminus B_i$. 
      Define
      \begin{equation}
        p_{ij} := \begin{cases}
          \overline p_{ij} + \hat p_{ij} & \text{if } j \in B_i \\
          \overline p_{ij} & \text{otherwise.}
        \end{cases}
      \end{equation}

      Let us now prove that (1) if $\OPT(I')>T$ then we have $\OPT(I)>T$, and (2) every schedule $\sigma$ with makespan $C^{I'}(\sigma)$ in $I'$ has a makespan
      at most $C^{I'}(\sigma)+T$ in $I$ ($C_{\Gamma}(\sigma) \le C^{I'}(\sigma)+T$).
     
      For (1), we prove that $\OPT(I) \le T$ implies that $\OPT(I') \le T$.
      Let $\sigma$ be an optimal solution of $I$ and $i$ a machine.
      $C_{\Gamma}(\sigma) \le T$ implies that $C_{\Gamma}(\sigma_i) \le T$ for any $i$, and thus that $\barp(\sigma_i)+\hatp_{\Gamma}(\sigma_i) \le T$.
      Now, observe that $B_i \subseteq \Gamma(\sigma_i)$. Indeed, assume towards contradiction that there exists $j \in B_i \setminus \Gamma(\sigma_i)$.
      This implies that $|\Gamma(\sigma_i)|=\Gamma$. As by definition, any $j' \in \Gamma(\sigma_i)$ has $\hatp_{ij'} \ge \hatp_{ij} > \frac{T}{\Gamma}$,
      we get that $\hatp_{\Gamma}(\sigma_i) > T$, a contradiction.
      This implies $C^{I'}(\sigma_i) = \barp(\sigma_i)+\hatp(B_i) \le \barp(\sigma_i)+\hatp_{\Gamma}(\sigma_i) \le T$.


      For (2), let $\sigma$ be a solution of $I'$.
      Let $i \in M$. Observe that $ \hatp(\Gamma(\sigma_i)) \le \hatp(B_i)+T$ as $\Gamma(\sigma_i)$ contains at most $\Gamma$ jobs in $\sigma_i \setminus B_i$,
      and these jobs have $\hatp_{ij} \le \frac{T}{\Gamma}$. Thus, $C_{\Gamma}(\sigma_i)=\barp(\sigma_i)+\hatp_{\Gamma}(\sigma_i) \le \barp(\sigma_i)+\hatp(B_i)+T=C^{I'}(\sigma_i)+T$.
      

    Thus, given a $T$ and $I$ we create $I'$ as above and run the $c$-approximation for $R||C_{\max}$ to get a solution $\sigma$.
    If $C^{I'}(\sigma) > cT$ then $\OPT(I')>T$, implying $\OPT(I) > T$, and thus we reject $T$.
    Otherwise, we consider $\sigma$ as a solution for $I$, and $C_\Gamma(\sigma) \le (c+1)T$.

    \qed \end{proof}

    Using the well-known 2-approximation algorithm from \cite{lenstra1990approximation}, we obtain immediately the following.
    \begin{corollary}
      There is a $3$-approximation for $R|\UG|C_{\max}$.
    \end{corollary}

    Since by this reduction uniform machines stay uniform we also obtain the following using the EPTAS of~\cite{DBLP:conf/icalp/JansenKV16} for the classical $Q||C_{\max}$ problem.
    \begin{corollary}
      For every $\epsilon > 0$ there is a
      $(2 + \epsilon)$-approximation for $Q|\UG|C_{\max}$ running in time $2^{O(1/\epsilon\log(1/\epsilon)^4)}+poly(n)$.
    \end{corollary}

    \section{A $2-\epsilon$ inapproximability for unrelated machines}
    \label{sec:inapprox}

    For the classical $R||C_{\max}$ problem, when all $p_{ij} \in \{1,\infty\}$, deciding if the optimal value is at most $1$ is polynomially solvable
    as it can be reduced to finding a matching in a bipartite graph.
    The result below shows that answering the same question for $R|\UG|C_{\max}$ is $\NP$-complete.

 \begin{theorem}
      \label{th-inapprox}
      Given an instance $I$ of $R|\UG|C_{\max}$, it is \NP-complete to decide if $\OPT(I) \le 1$ or $\OPT(I) \ge 2$, and thus for any $\epsilon >  0$ is no $(2-\epsilon)$-approximation algorithm for $R|\UG|C_{\max}$ unless $\P=\NP$, even for $\Gamma=1$ and when each job can be scheduled on at most $3$ machines.
 \end{theorem}
 
 \begin{proof}
Let us define a reduction from $3$-SAT to $R|\UG|C_{\max}$ with $\Gamma=1$.
Let $I_0$ be an instance of $3$-SAT with clauses $\{C_i, i \in [m_0]\}$ and variables $\{x_j, j \in [n_0]\}$. Each $C_i$ is of the form $l^1_i \vee l^2_i \vee l^3_i$ where $l^k_i \in \{x_j,\bar{x}_j\}$ for some $j$. We define an instance $I$ of $R|\UG|C_{\max}$ with $m=2n_0$ machines and $n=n_0+m_0$ jobs as follows. To each variable $x_j$ we associate two machines $\{j_f,j_t\}$. 
We create a set of $n_0$ variable jobs where for any $j \in [n_0]$, $\op_{j_fj}=\op_{j_tj}=1$, $\op_{i'j}=\infty$ for any other $i'$, and $\hp_{ij}=0$ for any $i \in [m]$.
For any clause $C_i$, $i \in [m_0]$ we define $M_i$: the set of $3$ machines corresponding to literals $\{l^k_i\}$ satisfying $C_i$.
For example, if $C_7 = x_1 \vee \bar{x}_3 \vee x_5$ then $M_7 = \{1_t,3_f,5_t\}$.
We now define a set of $m_0$ clause jobs as follows. For any $j \in [n_0+1,n_0+m_0]$, job $j$ represents clause $C_{j-n_0}$ with $\hp_{ij}=1$ iff $i \in M_{j-n_0}$, $\hp_{i'j}=\infty$ for any other $i'$, 
and $\op_{ij}=0$ for any $i\in [m]$.
For example, job $j=n_0+7$ is associated to $C_7$ where in particular $\hp_{1_tj}=\hp_{3_fj}=\hp_{5_tj}=1$.
Notice that each clause job can be scheduled on at most $3$ machines.
Let us now verify that $I_0$ is satisfiable iff $\OPT(I)=1$.

$\Rightarrow$.
Suppose $I_0$ is satisfied by assignment $a$. For any $j \in [n_0]$, we schedule $j$ on $j_t$ if $x_j$ is set to false in $a$ and on $j_f$ otherwise.
For any $j \in [n_0+1,n_0+m_0]$, we schedule job $j$ on any machine $i \in M_{j-n_0}$ corresponding to a literal satisfying $C_i$ in assignment $a$.
Notice that in this schedule, a machine either receives exactly one variable job, implying a makespan of $1$, or only clause jobs, also implying a makespan of $1$ as $\Gamma=1$. 

$\Leftarrow$.
Suppose that $\OPT(I)=1$ and let us define an assignment $a$. This implies that any variable job $j$ is either scheduled on machine $j_f$, in which case we set $x_j$ to true, 
or on machine $j_t$, in which case we set $x_j$ to false. As $\OPT(I)=1$, and clause job $j \in [n_0+1,n_0+m_0]$ is scheduled on a machine $i \in M_{j-n_0}$ that did not receive a variable job, 
implying that clause $j-n_0$ is satisfied by literal $i$.
 \qed \end{proof}

\section{A PTAS for identical machines}
\label{sec:PTAS}
Note that we can assume that $m < n$. If $m \ge n$, a trivial schedule
with every job on a different machine is optimal. In some problems the encoding length
may be much smaller than $m$, when $m$ is only encoded in binary. However, here a polynomial time
algorithm is allowed to have a polynomial dependency on $m$.

Recall that for the $P|\UG|C_{\max}$ problem, given two $n$ dimensional vectors $\hp$ and $\barp$ and the number of machine $m$,
the objective is to create a schedule $\sigma$ that minimizes $\max_{i\in M} C_{\Gamma}(\sigma_i)$.
Recall also that  $C_{\Gamma}(\sigma_i)=\barp(\sigma_i)+\hatp_{\Gamma}(\sigma_i)$, where $\barp(\sigma_i)= \sum_{j\in \sigma_i}\barp_{j}$,
and $\hatp_{\Gamma}(\sigma_i)$ is the sum of the $\hp_j$ values of the $\Gamma$ largest jobs (w.r.t. $\hp_j$) of $\sigma_i$ (or the sum of all $\hp_j$ values
if $|\sigma_i| \le \Gamma$).
To obtain a PTAS for $P|\UG|C_{\max}$, we will reduce to the following problem, which admits an EPTAS (see \cite{jansen2017eptas}).

\begin{problem}{\sc Unrelated Machines with few Machine Types}
\label{pb:rcmaxfewtypes}
  \begin{itemize}
  \item Input: $n$ jobs and a set $M$ of $m$ machines with processing times $p_{ij} \ge 0$ for job $j$
  on machine $i$. Moreover, there is a constant $k$ and machine types $T_1\dot\cup\cdots \dot\cup T_k = \{1,\dotsc,m\}$, such that every machine within a type behaves the same.
  Formally, for every $k'$, every $i, i'\in T_{k'}$ and every $j\le n$ it holds that $p_{ij}=p_{i'j}$ 
  \item Output: find a schedule $\sigma:J\rightarrow M$
  \item Objective function: minimize makespan $C(\sigma)=  \max_{i\in M} C(\sigma_i)$, where $C(\sigma_i) = \sum_{j\in\sigma_i} p_{ij}$
  \end{itemize}
\end{problem}

Notice that the EPTAS of~\cite{jansen2017eptas}
for this problem provides an $(1+\epsilon)$-approximation running in time $f(|I|,\epsilon,k)=2^{O(k\log(k)\frac{1}{\epsilon} \log^4(\frac{1}{\epsilon}))}+poly(|I|)$.

We also introduce the following decision problem.

\begin{problem}{\sc Unrelated Machines with few Machine Types and capacities}
  \label{prob:cap}
  \begin{itemize}
  \item Input: as above, but in addition every machine $i$ has a capacity $c_{i}\in(0,1]$.
Moreover, capacities are the same among a type (for any $k' \in [k]$, for any $i,i' \in T_{k'}$, $c_i = c_{i'}$) 
  \item Output: decide if there is a schedule where $C(\sigma_i)  \le c_i $ for any $i$. 
  \end{itemize}
\end{problem}

Notice that the EPTAS for Problem~\ref{pb:rcmaxfewtypes} allows to approximately decide Problem~\ref{prob:cap} in the following sense.
\begin{lemma}\label{lemma:eptasrcmaxcap}
  There is an algorithm that for any $\epsilon > 0$, either outputs a schedule with $C(\sigma_i) \le (1+\epsilon)\cdot c_i$ for any $i$,
  or reject the instance, proving that there is no schedule with $C(\sigma_i) \le c_i$ for any $i$. 
This algorithm runs in time $f(|I|,\epsilon,k)$ where $f$ is the complexity of the above EPTAS to get a $(1+\epsilon)$-approximation.
\end{lemma}

\begin{proof}
  Let $A$ be the EPTAS of~\cite{jansen2017eptas} for Problem~\ref{pb:rcmaxfewtypes}. 
Given a input $I$ of Problem~\ref{prob:cap} we define an input $I'$ of Problem~\ref{pb:rcmaxfewtypes} in the following way.
  For every $j\le n$, scale $p_{ij}$ to $p_{ij} / c_i$. Then, if $A(I') \le (1+\epsilon)$, we can convert the solution found by $A$ into a solution for $I$ of makespan at most $(1+\epsilon) \cdot c_i$ for any $i$.
  Otherwise, as $A$ is a $(1+\epsilon)$-approximation, it implies that $\OPT(I') > 1$, and thus that no solution can have makespan at most $c_i$ for any $i$.
\qed \end{proof}

Let us now describe the PTAS for $P|\UG|C_{\max}$. Our objective is to provide a $(1+O(\epsilon))$ dual approximation for $P|\UG|C_{\max}$. The constant factor with $\epsilon$ can be ignored, since we can divide $\epsilon$ with this constant
in the preprocessing.

\paragraph{1. Guess the makespan and scale $\OPT$ to $1$.}
Let $I$ be an input of $P|\UG|C_{\max}$, and $T$ be a positive value (representing the current threshold).
We start by redefining $I$ by scaling $p_j := \frac{p_j}{T}$.
Our objective is now to produce a schedule $\sigma$ with $C_{\Gamma}(\sigma) \le 1+\epsilon$, or to prove that $\OPT(I) > 1$.

\paragraph{2. Rounding deviations.}
Let us now define $I^1$ (having vectors $\overline p^1$ and $\hatp^1$) in the following way.
For any $j$, if $\hp_j < \epsilon/\Gamma$ then we set $\hp^1_j\leftarrow 0$. Intuitively, this
will only result in an error of at most $\Gamma \cdot \epsilon / \Gamma$ on
every machine. Otherwise ($\hp_j \ge \epsilon/\Gamma$), we define $\hp^1_j$ by rounding  $\hp_j$ to the closest smaller value of the form
$\epsilon/\Gamma \cdot (1 + \epsilon)^i$.

\begin{observation}\label{obs:I1}
In $I^1$ there are at most $O(1/\epsilon\log(\Gamma/\epsilon))$
deviation values, and at most $O(1/\epsilon\log(1/\epsilon))$ deviation values in the interval $[\epsilon/\Gamma, 1/\Gamma]$.
\end{observation}

In the following, we will denote by $C^{I'}_\Gamma(\sigma)$ the cost of $\sigma$ for instance $I'$.

\begin{observation}\label{obs:IversI1}
  If $\OPT(I) \le 1$ then $\OPT(I^1) \le 1$.
  If we get solution $\sigma^1$ of $I^1$, then $C^I_{\Gamma}(\sigma^1) \le (1+\epsilon)C^{I^1}_{\Gamma}(\sigma^1)+\epsilon$
\end{observation}

It only remains now to either produce a good solution of $I^1$ (of cost at most $1+O(\epsilon)$), or prove that $\OPT(I^1) > 1$.


\paragraph{3. Machine thresholds.}

Given any solution $\sigma$ of $I^1$ such that $C^{I^1}_{\Gamma}(\sigma) \le 1$, we can associate to $\sigma$ an outline $t = \o(\sigma)$ which is defined as follows.
For any machine $i$ with more that $\Gamma$ jobs, the threshold value $t_i$ is such that any 
job on $i$ with $\hp_j > t_i$ deviates (belongs to $\Gamma(\sigma_i)$) and none of the jobs with
$\hp_j < t_i$ deviate. Notice that among jobs with $\hp_j = t_i$, some may deviate,
but not necessarily all.
For any machine $i$ with at most $\Gamma$ jobs, we define $t_i = 0$, implying again that any
job  with $\hp_j > t_i$ deviates on $i$.
Notice that in both cases we have $\hp_{\Gamma}(\sigma_i) \ge \Gamma \cdot t_i$.
Notice also that $C^{I^1}_{\Gamma}(\sigma) \le 1$ implies $t_i \le \frac{1}{\Gamma}$. Indeed, if we had $t_i > \frac{1}{\Gamma}$, there would be $\Gamma$ deviating jobs  with $\hp_j > t_i$, implying $C^{I^1}_{\Gamma}(\sigma_i) > 1$, a contradiction.
Let us denote by $\Delta$ the set of all possible values of a $t_i$. According to Observation~\ref{obs:I1} we have $|\Delta| =O(1/\epsilon\log(1/\epsilon))$.
Let $\T = \Delta^m$ be the set of all outlines (of solutions of cost at most $1$).

\begin{lemma}\label{lemma:guessoutline}
  Consider a solution $\sigma^{1*}$ of $I^1$ such that $C_\Gamma(\sigma^{1*}) \le 1$, and let $t^* = \o(\sigma^{1*})$.
  Then, we can guess in $m^{O(1/\epsilon\log(1/\epsilon))}$ time the vector $t^*$ (or a permutation thereof).
\end{lemma}
\begin{proof}
As $t^* \in T$, all the $t^*_i$ have a value in $\{0\}\cup [\frac{\epsilon}{\Gamma},\frac{1}{\Gamma}]$.
Thus, as deviating values are rounded in $I^1$, there are only a constant number of possible
threshold value and we can guess them. For every possible threshold, we guess how many machines in the optimal
solution have it.
\qed \end{proof}

Thus, we can now assume that we know the vector $t^*$.

\paragraph{4. Constructing an instance with few machine types and capacities.}
To give an insight of the correct reduction defined below, let us first see what happen if we define an instance $I^2(t^*)$ of $R||C_{\max}$ as follows.
For simplicity, we also assume that there are no job with $\hp_j = t^*_i$ on each machine $i$ in the previously considered optimal solution of $I^1$.
For any machine $i$ and job $j$, define the processing time in $I^2(t^*)$ as $p_{ij}=\op_j+\hp_j$ if $\hp_j \ge t^*_i$, and $p_{ij}=\op_j$ otherwise.
Then, consider the following implications.
\begin{enumerate}
\item if $OPT(I^1) \le 1$, then $\OPT(I^2(t^*)) \le 1$
\item for any solution $\sigma'$ of $I^2(t^*)$, $C^{I^1}_{\Gamma}(\sigma') \le C^{I^2}(\sigma')$ (implying that if there exists $\sigma'$ with $C^{I^2}(\sigma') \le 1+\epsilon$, then we will have our solution for $I^1$ of cost $1+\epsilon$)
\end{enumerate}
While Property $(1)$ holds, this is not the case for Property $(2)$. Indeed, suppose that in $\sigma'$ there is a machine $i$ such that for all jobs $j$ scheduled on $i$, $\hp_j < t^*_i$. This implies that $C(\sigma_i)=\sum_{j \in \sigma_i} \op_j$.
However, if we look now at $\sigma'$ in $I^1$, we get $C^{I^1}_{\Gamma}(\sigma_i)=C^{I^2}(\sigma_i) + \hp(\Gamma(\sigma_i))$, which is greater than the claimed value.
To solve this problem we have to remember in $R||C_{\max}$ that there will be a space of size at most $\Gamma \cdot t_i$ which will be occupied by deviations.

Let us now turn to the correct version.
\begin{definition}\label{def:I2}
For any $t \in \T$, we define the following input $I^2(t)$ of Problem~\ref{prob:cap}.
We set the machine capacity to
\begin{equation*}
  c_i := 1 - \Gamma \cdot t_i + \epsilon .
\end{equation*}
The addition of $\epsilon$ is only a technicality to ensure that all $c_i$ are non-zero.
Note that if there are less than $\Gamma$ jobs on $i$, then
$t_i$ must be $0$ and therefore $c_i = 1 + \epsilon$.
For every job $j$ set
\begin{equation*}
  p_{ij} := \begin{cases}
    \op_j + \hp_j - t_i &\text{ if } \hp_j \ge t_i, \\
    \op_j &\text{ if } \hp_j < t_i. \\
  \end{cases}
\end{equation*}
\end{definition}
Note that at $\hp_j = t_i$, the values of both cases are equal. Notice also that in $I^2(t)$ there are only $|\Delta|$ different machine types.

\begin{lemma}\label{lemma:I1versI2}
If $\OPT(I^1) \le 1$ and $t$ is the outline of an optimal solution $\sigma^2$,
for any $i$, $C^{I^2(t)}(\sigma^2_i) \le c_i$.
\end{lemma}
\begin{proof}
  Let us consider jobs $\sigma^2_i$ scheduled on machine $i$.
  If $t_i = 0$, then
  \begin{align*}
    \sum_{j\in \sigma^2_i} p_{ij} = \sum_{j\in \sigma^2_i} \op_j + \hp_j \le 1 < c_i .
  \end{align*}
  Assume now $t_i > 0$, implying that $|\Gamma(\sigma^2_i)| \ge \Gamma$.
    By choice of $t_i$, every job $j\in \Gamma(\sigma^2_i)$ has $\hp_j\ge t_i$
  and every $j\in \sigma^2_i\setminus \Gamma(\sigma^2_i)$ has $\hp_j\le t_i$.
  This implies
  \begin{align*}
    \sum_{j\in \sigma^2_i} p_{ij} = \sum_{j\in \Gamma(\sigma^2_i)} p_{ij} + \sum_{j\in \sigma^2_i\setminus \Gamma(\sigma^2_i)} p_{ij}
    = \sum_{j\in \Gamma(\sigma^2_i)} [\op_j + \hp_j - t_i] + \sum_{j\in \sigma^2_i\setminus \Gamma(\sigma^2_i)} \op_j
    \le 1 - \Gamma \cdot t_i < c_i .
  \end{align*}
\qed \end{proof}
\begin{lemma}\label{lemma:I2versI1}
  For any $t \in \T$, if there is a solution $\sigma^2$ of $I^2(t)$ such that $C^{I^2(t)}(\sigma^2_i) \le (1+\epsilon)\cdot c_i$ for any $i$,
  then $C^{I^1}_{\Gamma}(\sigma^2) \le (1 + \epsilon)^2$.
\end{lemma}
\begin{proof}
Let $i$ be a machine.
Then for every $j\in \sigma^2_i$,
\begin{equation*}
  p_{ij} = \begin{cases}
    \op_j + \hp_j - t_i \ge \op_j &\text{ if } \hp_j \ge t_i, \\
    \op_j &\text{ if } \hp_j < t_i.
  \end{cases}
\end{equation*}
Furthermore, for every $j\in \Gamma(\sigma^2_i)$,
\begin{equation*}
  p_{ij} = \begin{cases}
    \op_j + \hp_j - t_i &\text{ if } \hp_j \ge t_i, \\
    \op_j > \op_j + \hp_j - t_i &\text{ if } \hp_j < t_i.
  \end{cases}
\end{equation*}
This implies,
\begin{multline*}
  \sum_{j\in \Gamma(\sigma^2_i)} [\op_j + \hp_j] + \sum_{j\in \sigma^2_i\setminus \Gamma(\sigma^2_i)} \op_j
  \le \Gamma \cdot t_i + \sum_{j\in \Gamma(\sigma^2_i)} [\op_j + \hp_j - t_i] + \sum_{j\in \sigma^2_i\setminus \Gamma(\sigma^2_i)} \op_j \\
  \le \Gamma \cdot t_i + \sum_{j\in \Gamma(\sigma^2_i)} p_{ij} + \sum_{j\in \sigma^2_i\setminus \Gamma(\sigma^2_i)} p_{ij} = \Gamma \cdot t_i + \underbrace{\sum_{j\in \sigma^2_i} p_{ij}}_{\le (1 + \epsilon)\cdot c_i} \le \Gamma\cdot t_i + (1 + \epsilon)\cdot (1 - \Gamma\cdot t_i + \epsilon) \le (1 + \epsilon)^2 .
\end{multline*}
\qed \end{proof}

\begin{theorem}\label{thm:ptas}
There is a $(1+\epsilon)$-approximation algorithm for $P|\UG|C_{\max}$ running in time $O(m^{O(1/\epsilon\log(1/\epsilon))} \times f(|I|,\epsilon,O(1/\epsilon\log(1/\epsilon)))$ where $f$ is the function of Lemma~\ref{lemma:eptasrcmaxcap}.
\end{theorem}
\begin{proof}
  Given $I$ input of $P|\UG|C_{\max}$ and a threshold $T$, we run algorithm $A$ of Lemma~\ref{lemma:eptasrcmaxcap} on $I^2(t)$ for any $t \in \T$ with a precision $\epsilon$.
  If $A$ rejects all the $I^2(t)$ then we can reject $T$ according to Observation~\ref{obs:IversI1} and Lemma~\ref{lemma:I1versI2}.
  Otherwise, there exists $t_0$ such that $A(I^2(t_0))$ outputs a schedule $\sigma^2$ where $C^{I^2(t_0)}(\sigma^2) \le (1+\epsilon)\cdot c_i$ for any $i$,
  implying $C^I_\Gamma(\sigma^2) \le (1+\epsilon)C^{I^1}_\Gamma(\sigma^2)+\epsilon \le (1+\epsilon)^3+\epsilon \le 1+5\epsilon$ according to Observation~\ref{obs:IversI1} and Lemma~\ref{lemma:I2versI1} (for sufficiently small $\epsilon$).
Finally, the running time is as claimed due to the bound of $\T$ in Lemma~\ref{lemma:guessoutline}.
\qed \end{proof}

\section{EPTAS for identical machines}
\label{sec:EPTAS}

The approach for an EPTAS is similar to the PTAS above. We would like to remove the bottleneck from the previous
section, which is the guessing the thresholds.
In the PTAS we notice that even if the thresholds were chosen incorrectly, but we find a solution
to the derived problem, we can get a good solution for the initial problem.
Informally, we will now still create an instance of Problem~\ref{prob:cap}, but we only guess approximately the number of machines for each threshold.

We start by defining $I^1$ as in the previous section.
Given any solution $\sigma^1$ of $I^1$ such that $C^{I^1}_{\Gamma}(\sigma^1) \le 1$, we can associate to $\sigma$ a restricted outline $\m=\ro(\sigma)$ where $\m$ is defined as follows.
Let $t=\o(\sigma)$.
For any threshold value $l \in \Delta$, let $m_l= |\{i | t_i = l\}|$ be the number of machines with threshold $l$ in $\sigma^1$.
We define $\m_l \in \{0, 1, 2, 4, 8,\dotsc, 2^{\lfloor \log(m)\rfloor}\}$ such that $\m_l \le m_l < 2 \m_l$.
Let $\rT = \{\m\in\{0, 1, 2, 4, 8,\dotsc, 2^{\lfloor \log(m)\rfloor}\}^\Delta \mbox{ such that $\frac{m}{2} \le \sum_l \m_l\le m$} \}$ be the set of restricted outlines (of solutions of cost at most $1$).

\begin{lemma}\label{lemma:guessrestrictedoutline}
  Consider a solution $\sigma^{1*}$ of $I^1$ such that $C_\Gamma(\sigma^{1*}) \le 1$, and let $\m^* = \ro(\sigma^{1*})$.
  Then, we can guess in  time $2^{O(1/\epsilon\log^2(1/\epsilon))} + m^{O(1)}$ the vector $\m^*$.
\end{lemma}
\begin{proof}
Clearly it suffices to iterate over all values $\m_i\in \{0, 1, 2, 4, 8,\dotsc, 2^{\lfloor \log(m)\rfloor}\}$, i.e., $O(\log(m))$ many.
Guessing this number for every threshold value in $\Delta$ takes $\log^{O(1/\epsilon\log(1/\epsilon))}(m)$ time.
Consider first the case when $\log(m)/\log\log(m) \le 1/\epsilon\log(1/\epsilon)$.
For sufficiently large $m$ it holds that $\log^{1/2}(m) \le \log(m)/\log\log(m) \le 1/\epsilon\log(1/\epsilon)$. Hence,
\begin{equation*}
  \log^{O(1/\epsilon\log(1/\epsilon))}(m)
  = (\log^{1/2}(m))^{2 \cdot O(1/\epsilon\log(1/\epsilon))}
  \le (1/\epsilon\log(1/\epsilon))^{O(1/\epsilon\log(1/\epsilon))}
  \le 2^{O(1/\epsilon\log^2(1/\epsilon))} .
\end{equation*}
If on the other hand $\log(m)/\log\log(m) \ge 1/\epsilon\log(1/\epsilon)$, then
\begin{equation*}
  \log^{O(1/\epsilon\log(1/\epsilon))}(m)
  \le \log^{O(\log(m)/\log\log(m))}(m)
  = 2^{O(\log(m)/\log\log(m)\cdot\log\log(m))}
  = m^{O(1)} .
\end{equation*}
We conclude,
\begin{equation*}
  \log^{O(1/\epsilon\log(1/\epsilon))}(m)
  \le 2^{O(1/\epsilon\log^2(1/\epsilon))} + m^{O(1)} .
\end{equation*}
From all the guesses, we report fail whenever $\sum_i m_i < m/2$ or $\sum_i m_i > m$.
\qed \end{proof}

For any $\m \in \rT$, we define the following input $I^2(\m)$ of Problem~\ref{prob:cap}.
We first create for any $l$ a set $M_l$ of $\m_l$ machines where for each machine $i \in M_l$ the capacity and the $p_{ij}$ are defined as in Definition~\ref{def:I2} for threshold $t_i=l$.
Then, we create another set $M'_l$ of $\m_l$ machines (that we call cloned machines) with the same capacity and the same $p_{ij}$ values.
Let $m' = \sum \m_l$. Notice that the total number of machines is $2m'$, with $m \le 2m' < 2m$.
Thus, we have to ensure that not too many machines are used in total.
For that purpose we add a set of $2m' - m$ dummy jobs $D$, where all $j \in D$ have $p_{ij}=\infty$ on the
original machines $i \in M_l$ and $p_{ij}=c_i$ on every cloned machine $i \in M'_l$. 
Notice that the number of types is now $2|\Delta|$, which is still small enough to get an EPTAS.
Let us call the non-dummy jobs \emph{regular} jobs. 

\begin{lemma}\label{lemma:I1versI2EPTAS}
If $\OPT(I^1) \le 1$ and $\m$ is the restricted outline of an optimal solution, then there exists a solution $\sigma^2$ of $I^2(\m)$ such that for any $i$, $C^{I^2(\m)}(\sigma^2_i) \le c_i$.
\end{lemma}

\begin{proof}
Let $m^*_l= |\{i | t_i = l\}|$ be the number of machines with threshold $l$ in the considered optimal solution of $I^1$.
Let $l \in \Delta$ be a threshold value. We first schedule  $2\m_l - m^*_l$ many dummy jobs on
cloned machines of $M'_l$. This will cover all dummy jobs, since
\begin{equation*}
  \sum_l[2 \m_l - m^*_l] = 
  2\sum_l \m_l - \sum_l m^*_l = 
  2m' - m .
\end{equation*}
We will now schedule all remaining jobs on the empty machines. For every
threshold value $l$ we have $2 \m_l - (2\m_l - m^*_l) = m^*_l$ many empty machines.
In other words, we are left with an instance with the exact same number of machines
for each threshold as in the optimal solution and with the original jobs.
As argued in Lemma~\ref{lemma:I1versI2}, we get the desired claim.
\qed \end{proof}

\begin{lemma}\label{lemma:I2versI1EPTAS}
  For any $\m \in \rT$, if there is a solution $\sigma^2$ of $I^2(\m)$ such that $C^{I^2(\m)}(\sigma_i^2) \le c_i+\epsilon$ for any $i$,
  then we can deduce a solution $\sigma^3$ for $I^1$ with $C^{I^1}_{\Gamma}(\sigma^3) \le (1 + 2\epsilon)^2$.
\end{lemma}

\begin{proof}
We will first normalize $\sigma^2$. Since dummy jobs have $p_{ij}=c_i$ on cloned machines,
in a $(1+\epsilon)$-approximation there can only be one per machine (assuming that $\epsilon < 1$).
Indeed, there may still be a load of $\epsilon \cdot c_i$ from other jobs on
the same machine. We want to ensure that every machine either has a dummy job or some regular load,
but not both. For every threshold value $l \in \Delta$, there can be at most $\m_l$ machines in $M'_l$ that
have a dummy job. For any such machine in $M'_l$, we remove all the regular jobs (of total load of at most $\epsilon \cdot c_i$) from it and move them
to one of the original machines in $M_l$, without using the same machine in $M_l$ twice.
Since for any $i \in M_l$ we had $C^{I^2(\m)}(\sigma^2_i) \le (1 + \epsilon) c_i$ before moving the jobs,
and since regular jobs have the same processing time on machines $M_l$ and $M'_l$,
after moving the jobs we get $C^{I^2(\m)}(\sigma^2_i) \le (1+2\epsilon) c_i$ for any $i \in M_l$.
We now have a solution violating the capacities
by at most $2\epsilon \cdot c_i$ such that a machine with a dummy job has no other jobs.

We now forget about all dummy jobs and the machines they are on.
What we are left with is a set of $m$ machines (with some thresholds $t$) such that for any $i$ we have $C^{I^2(\m)}(\sigma^2_i) \le (1+2\epsilon) c_i$.
By Lemma~\ref{lemma:I2versI1} we get the desired result.
\qed \end{proof}

All in all, we were able to reduce the number of instances created to 
only $2^{O(1/\epsilon\log^2(1/\epsilon))} + m^{O(1)}$ many and removed the bottleneck from the PTAS this way.
As in Theorem~\ref{thm:ptas}, given an instance of $P|\UG|C_{\max}$ we will use the algorithm of Lemma~\ref{lemma:eptasrcmaxcap} on $I^2(m)$ for any $m \in \rT$. This leads to the following result.

\begin{theorem}\label{thm:eptas}
  There is a $(1+\epsilon)$-approximation algorithm for $P|\UG|C_{\max}$ running in time $O(2^{O(1/\epsilon\log^2(1/\epsilon))} + m^{O(1)}) \times f(|I|,\epsilon,O(1/\epsilon\log(1/\epsilon)))$ where $f$ is the function of Lemma~\ref{lemma:eptasrcmaxcap}.
\end{theorem}

 \bibliographystyle{abbrv}
 \bibliography{ref}

\end{document}